%% file: main.tex
\documentclass[aps,pra,preprint,superscriptaddress,longbibliography,nofootinbib]{revtex4-1}
\usepackage[utf8]{inputenc}
\usepackage{amsfonts}
\usepackage{amsmath}
\usepackage{amssymb}
\usepackage{amsthm}
\usepackage{mathtools}
\usepackage{enumerate}
\usepackage[shortlabels]{enumitem}
\usepackage{array}
\usepackage{bm}
\usepackage{graphicx}
\usepackage{nicefrac}
\usepackage{booktabs}
\usepackage[all]{xy}
\usepackage{hyperref}
\usepackage{subcaption}
\usepackage{bbm}
\usepackage{tikz}
\usetikzlibrary{arrows}
\usetikzlibrary{positioning}
\usepackage{cancel}
\usepackage{xcolor}
\usepackage[font=small,labelfont=bf]{caption}
\usepackage{MnSymbol}
\usepackage{multirow}

\newcommand\mc[1]{\mathcal{#1}}
\newcommand\cH{\mc{H}}
\newcommand\cK{\mc{K}}
\newcommand\BH{\mc{B}(\cH)}

\newcommand\PH{\mc{P}(\cH)}
\newcommand\PV{\mc{P}(V)}

\newcommand\cN{\mc{N}}

\newcommand\cJ{\mc{J}}
\newcommand\SN{\mc{S}(\cN)}
\newcommand\PN{\mc{P}(\cN)}

\newcommand\VN{\mc{V}(\mc{N})}

\newcommand\JN{\mc{J}(\mc{N})}
\newcommand\JNsa{\mc{J}(\mc{N})_\mathrm{sa}}

\newcommand\PPi{\underline{\Pi}}

\newcommand\ra{\rightarrow}

\newcommand\hra{\hookrightarrow}

\newcommand\R{\mathbb{R}}
\newcommand\C{\mathbb{C}}

\newtheorem{theorem}{Theorem}
\newtheorem{lemma}{Lemma}
\newtheorem{definition}{Definition}

\begin{document}
\title{Gleason's theorem for composite systems}
\author{Markus Frembs}
\email{m.frembs@griffith.edu.au}
\affiliation{Centre for Quantum Dynamics, Griffith University,\\ Yugambeh Country, Gold Coast, QLD 4222, Australia}
\vspace{-0.5cm}
\author{Andreas D\"oring}
\email{andreas.doering@posteo.de}

\begin{abstract}
    Gleason's theorem \cite{Gleason1975} is an important result in the foundations of quantum mechanics, where it justifies the Born rule as a mathematical consequence of the quantum formalism. Formally, it presents a key insight into the projective geometry of Hilbert spaces, showing that finitely additive measures on the projection lattice $\PH$ extend to positive linear functionals on the algebra of bounded operators $\BH$. Over many years, and by the effort of various authors, the theorem has been broadened in its scope from type I to arbitrary von Neumann algebras (without type $\text{I}_2$ factors). Here, we prove a generalisation of Gleason's theorem to composite systems. To this end, we strengthen the original result in two ways: first, we extend its scope to dilations in the sense of Naimark and Stinespring \cite{Naimark1943,Stinespring1955} and second, we require consistency with respect to dynamical correspondences on the respective (local) algebras in the composition \cite{AlfsenShultz1998}. We show that neither of these conditions changes the result in the single system case, yet both are necessary to obtain a generalisation to bipartite systems.
\end{abstract}

\maketitle


\section{Gleason's theorem}

Gleason's theorem is a landmark result in quantum theory. It justifies the Born rule, which originally had the status of an axiom, as a mathematical fact of the projective geometry of Hilbert spaces.

\begin{theorem}\label{thm: original Gleason}
    \textbf{\emph{(Gleason \cite{Gleason1975})}}
    Let $\cH$ be a separable Hilbert space, $\mathrm{dim}(\cH) \geq 3$.
    Then every countably additive probability measure $\mu: \mc{P}(\mc{H}) \rightarrow [0,1]$ over the projections on $\mc{H}$ is of the form $\mu(p) = \mathrm{tr}[\rho_\mu p]$ for all $p \in \mc{P}(\mc{H})$, with $\rho_\mu: \mc{H} \rightarrow \mc{H}$, $\rho_\mu \geq 0$, $\mathrm{tr}[\rho_\mu] = 1$ a density operator.
\end{theorem}

Here, a \emph{countably additive probability measure} is a map $\mu: \mc{P}(\mc{H}) \rightarrow [0,1]$, $\mu(1) = 1$ and such that $\mu(\sum_{i=1}^\infty p_i) = \sum_{i=1}^\infty \mu(p_i)$ for all $p_i \in \PH$ such that $p_ip_j=0$ whenever $i\neq j$. Note that when $\mc{N} = M_2(\mathbb{C})$, there exist measures that fail to extend to linear functionals.

While the original argument was for type I factors only, Thm.~\ref{thm: original Gleason} was later extended to type II and type III von Neumann algebras in \cite{Christensen1982, Yeadon1983,Yeadon1984} (see also \cite{Maeda1989}). In this setting, $\mu$ is further called completely additive if $\mu(\sum_{i \in I} p_i) = \sum_{i \in I} \mu(p_i)$ for every family of pairwise
orthogonal projections $(p_i)_{i \in I}$, $p_i \in \PN$.\footnote{For $\cN = \BH$ with $\cH$ separable, complete additivity reduces to countable additivity.} Recall that a state $\sigma \in \SN$ is a positive, normalised linear functional, it is \emph{normal} if $\sigma(\sum_{i \in I} p_i) = \sum_{i \in I} \sigma(p_i)$ for all families of pairwise orthogonal projections $(p_i)_{i \in I}$, $p_i \in \PN$ (Thm. 7.1.12 in \cite{KadisonRingroseII}).
Clearly, a (normal) state $\sigma$ defines a finitely (completely) additive probability measure on $\mc{P}(\cN)$, conversely:

\begin{theorem}\label{thm: general Gleason}
    \textbf{\emph{(Gleason-Christensen-Yeadon \cite{Christensen1982,Yeadon1983,Yeadon1984})}} Let $\cN$ be a von Neumann algebra with no summand of type $I_2$ and let $\mu: \mc{P}(\cN) \rightarrow \mathbb{R}$ be a finitely additive probability measure on the projections of $\cN$. There exists a unique state $\sigma_\mu \in \mc{S}(\cN)$ such that $\mu(p) = \sigma_\mu(p)$ for all $p \in \mc{P}(\cN)$. If $\mu$ is completely additive then $\sigma_\mu$ is normal and of the form $\mu(p) = \sigma_\mu(p) = \mathrm{tr}[\rho_\mu p]$ for all $p \in \mc{P}(\cN)$ with $\rho_\mu$ a positive trace-class operator. 
\end{theorem}



\section{Gleason's theorem in context}

\subsection{Partial order of contexts and probabilistic presheaf.}

Note that finite (complete) additivity 
only assumes that $\mu: \PN \ra [0,1]$ is \emph{quasi-linear}, i.e., linear in commuting von Neumann subalgebras. Remarkably, Thm.~\ref{thm: original Gleason} and its generalisation Thm.~\ref{thm: general Gleason} show that this already implies linearity on all of $\cN$.

We emphasise the passage from quasi-linearity to linearity as follows. Note that $\mu$ restricts to a probability measure $\mu_V$ in every commutative von Neumann subalgebra $V \subset \cN$, i.e., $\mu_V = \mu|_V$. $\mu$ thus defines a collection of probability measures $(\mu_V)_{V \subset \cN}$, one for every commutative von Neumann subalgebra, and such that whenever $\tilde{V} \subset V$ is a von Neumann subalgebra, the respective probability measures are related by restriction, $\mu_{\tilde{V}} =\mu_V|_{\tilde{V}}$ (viz. marginalisation).

From this perspective, the constraints on measures $\mu: \PN \ra [0,1]$ in Gleason's theorem arise via the inclusion relations between commutative von Neumann subalgebras. This motivates the definition of the following partially ordered set (poset).

\begin{definition}\label{def: context category}
    Let $\cN$ be a von Neumann algebra. The poset of commutative von Neumann subalgebras of $\cN$ is called the \emph{context category of $\cN$} and is denoted by $\VN$.
\end{definition}

The name `context category' is motivated as follows. First, note that every poset can be regarded as a category, whose objects are the elements of the poset and whose arrows are defined by $a\ra b \Leftrightarrow a\leq b$. Second, in quantum physics contextuality refers to the fact that not all observables, represented by the self-adjoint operators in $\cN$, can be measured simultaneously in an arbitrary state.\footnote{This is in contrast to classical physics, where observables are represented by elements in a (commutative) algebra of functions on a (locally) compact Hausdorff space. The Kochen-Specker theorem shows that in the noncommutative case such a description
is impossible \cite{KochenSpecker1967}.} Only commuting operators can be measured simultaneously, i.e., those that are contained in a commutative von Neumann subalgebra---a \emph{context} \cite{DoeringFrembs2019a}. The constraints in Gleason's theorem are therefore \emph{noncontextuality constraints}: a projection $p \in \PN$ is assigned a probability $\mu(p)$ independent of the context that $p$ lies in: $\mu_V(p) = \mu_{\tilde V}(p) = \mu(p)$ whenever $p\in \tilde{V},V$.

Next, we formalise the idea that a measure $\mu: \PN \ra [0,1]$ corresponds to a collection of probability measures over $\VN$.

\begin{definition}\label{defn: probabilistic presheaf}
    Let $\cN$ be a von Neumann algebra with context category $\VN$. The \emph{(normal) probabilistic presheaf\footnote{A presheaf $\underline{P}$ over the category $\mc{C}$ is a functor $\underline{P}: \mc{C}^\mathrm{op} \ra \mathbf{Set}$. We denote presheaves with an underscore.} $\PPi$ of $\cN$ over $\VN$} is the presheaf given
    \begin{itemize}
    	\item [(i)] on objects: for all $V\in\VN$, let
    	\begin{equation*}
    			\PPi_V:=\{\mu_V:\PV\ra [0,1] \mid \mu_V\text{ is a finitely (completely) additive probability measure}\}\; ,
    	\end{equation*}
    	\item [(ii)] on arrows: for all $\tilde{V},V\in\VN$ with $\tilde{V} \subset V$, let $\PPi(i_{\tilde{V}V}): \PPi_V \ra \PPi_{\tilde{V}}$ with $\mu_V \mapsto \mu_V|_{\tilde{V}}$, where $i_{\tilde{V} V}:\tilde{V}\hra V$ denotes the inclusion map between contexts $\tilde{V} \subset V$. 
    \end{itemize}
\end{definition}

We remark that the study of presheaves over the partial order of contexts is at the heart of the topos approach to quantum theory \cite{IshamButterfieldI,IshamButterfieldII,IshamButterfieldIII,IshamButterfieldIV,IshamDoeringI,IshamDoeringII,IshamDoeringIII,IshamDoeringIV,HLS2009,HLS2010}. For an introduction, see e.g. \cite{DoeIsh11}. For the intimate relationship between contextuality (in the sense of Def.~\ref{def: context category}) and various key theorems in quantum theory, see \cite{DoeringFrembs2019a}.

\subsection{Gleason's theorem in presheaf form.}

We have seen that a finitely (completely) additive probability measure over the projections of a von Neumann algebra $\cN$ can be regarded as a collection of finitely (completely) additive probability measures over contexts $(\mu_V)_{V \in \VN}$, which satisfy the constraints $\mu_{\tilde{V}} = \mu_V|_{\tilde{V}}$ whenever $\tilde{V},V \in \VN$ with $\tilde{V} \subset V$. In terms of Def.~\ref{defn: probabilistic presheaf}, $\mu$ thus becomes a \emph{global section} of the (normal) probabilistic presheaf $\PPi$. We denote the set of global sections of $\PPi$ by $\Gamma[\PPi].$\footnote{Here, `global' refers to the fact that $\mu$ satisfies the restriction constraints in $\PPi$ over all of $\VN$. In contrast, a \emph{local section} satisfies the constrains only on a sub-poset of $\VN$.} In this terminology, we obtain the following reformulation of Gleason's theorem \cite{Doering2004,DeGroote2007,Doering2012}.

\begin{theorem}\label{thm: Gleason in context form}
    \textbf{\emph{(Gleason in presheaf form (I))}} Let $\cN$ be a von Neumann algebra with no summand of type $\text{I}_2$. There is a bijective correspondence between (normal) states on $\cN$ and global sections of the (normal) probabilistic presheaf $\PPi$ of $\cN$ over $\VN$.
\end{theorem}

\begin{proof}
    Note that every (normal) state $\sigma: \cN \ra \C$ defines a finitely (completely) additive probability measure over the projections of $\cN$ (cf. Thm. 7.1.12 in \cite{KadisonRingroseII}), which corresponds with a global section of $\PPi$ by the preceeding discussion. By the same correspondence, each global section of the (normal) probabilistic presheaf $\gamma \in \Gamma[\PPi(\VN)]$ extends to a (normal) state $\sigma \in \SN$ as a consequence of Thm.~\ref{thm: general Gleason}.
\end{proof}


Next, we turn to a generalisation of Gleason's theorem in presheaf form to composite systems.

\subsection{Linearity without positivity.}

The canonical product on partial orders, denoted $\mc{V}_1 \times \mc{V}_2$, is the Cartesian product with elements $(V_1,V_2)$ for $V_1 \in \mc{V}_1$, $V_2 \in \mc{V}_2$ and order relations such that for all $\tilde{V}_1,V_1 \in \mc{V}_1$, $\tilde{V}_2,V_2 \in \mc{V}_2$:
\begin{equation*}\label{eq: product context category}
    (\tilde{V}_1,\tilde{V}_2) \subseteq (V_1,V_2) :\Longleftrightarrow \tilde{V}_1 \subseteq_1 V_1 \mathrm{\ and \ } \tilde{V}_2 \subseteq_2 V_2\; .
\end{equation*}
It is interesting to ask whether the state space of the composite system can be recovered from product contexts $V \in \mc{V}(\cN_1) \times \mc{V}(\cN_2) \subsetneq \mc{V}(\cN_1 \otimes \cN_2)$ only. For this to be possible states on the subsystem algebras need to define unique states on the composite algebra and vice versa. For this reason we will consider the spatial tensor product between von Neumann algebras $\bar{\otimes}$ (for details, see section 11.2 in \cite{KadisonRingroseII}). Recall that given normal states $\sigma_1 \in (\cN_1)_*$ and $\sigma_2 \in (\cN_2)_*$, where $\cN_*$ denotes the predual of $\cN$, there exists a unique normal product state $\sigma = \sigma_1 \bar{\otimes} \sigma_2 \in (\cN_1 \bar{\otimes} \cN_2)_*$ (Prop. 11.2.7, \cite{KadisonRingroseII}). Conversely, every normal linear functional $\sigma \in \cN_*$ on the spatial tensor product $\cN = \cN_1 \bar{\otimes} \cN_2$ is the norm limit of normal product states 
$\sigma_1 \bar{\otimes} \sigma_2$ (Prop 11.2.8, \cite{KadisonRingroseII}). 
It follows that we can identify the product context $V = (V_1,V_2) \in \mc{V}(\cN_1) \times \mc{V}(\cN_2)$ with the commutative von Neumann subalgebra $V = V_1 \bar{\otimes} V_2 \subset \cN_1 \bar{\otimes} \cN_2$ for every $V_1 \subset \cN_1$ and $V_2 \subset \cN_2$.

\begin{theorem}\label{thm: linearity vs positivity}
    Let $\cN_1$, $\cN_2$ be von Neumann algebras with no summand of type $\text{I}_2$ and let $\cN = \cN_1 \bar{\otimes} \cN_2$. There is a bijective correspondence between global sections of the normal probabilistic presheaf $\PPi$ of $\cN$ over $\mc{V}(\cN_1) \times \mc{V}(\cN_2)$ and normal linear functionals $\sigma: \cN \ra \C$ such that $\sigma(1) = 1$ and $\sigma(a \otimes b) \geq 0$ for all $a \in (\cN_1)_+$ and $b \in (\cN_2)_+$.
\end{theorem}

\begin{proof}
    \input{Proof_Linearity}
\end{proof}

We emphasise that the reduction to products of commutative von Neumann algebras in $\mc{V}(\cN_1) \times \mc{V}(\cN_2)$ is already sufficient to lift quasi-linear measures to linear functionals. In finite dimensions this has been observed before \cite{KlayRandallFoulis1987,Wallach2002}. Yet, unlike Thm.~\ref{thm: Gleason in context form} the linear functionals thus obtained are not necessarily positive, hence, Thm.~\ref{thm: linearity vs positivity} does not provide a classification of states on $\cN = \cN_1 \bar{\otimes} \cN_2$. From this perspective, it seems unjustified to call Thm.~\ref{thm: linearity vs positivity} a generalisation of Gleason's theorem (Thm.~\ref{thm: general Gleason}) to composite systems.

\section{A composite Gleason theorem}

Thm.~\ref{thm: linearity vs positivity} makes it clear that in order to obtain a generalisation of Gleason's theorem, which singles out the state spaces of composite systems, one needs to strengthen Def.~\ref{defn: probabilistic presheaf}. In particular, one needs to characterise positivity of the linear functional $\sigma: \cN_1 \bar{\otimes} \cN_2 \ra \C$ in Thm.~\ref{thm: linearity vs positivity}. We will derive positivity of $\sigma$ from complete positivity of an associated map $\phi: \cN_1 \ra \cN^*_2$ (where $\cN_2^*$ denotes the algebra under the adjoint, see Sec.~\ref{sec: time orientation} below).

Recall that in the proof of Thm.~\ref{thm: linearity vs positivity} we constructed a map $\widetilde{\phi}: \cN_1 \ra (\cN_2)_*$, into normal states on $\cN_2$ (up to normalisation). We may therefore identify the normal states $\widetilde{\phi}(p) \in (\cN_2)_*$ for every $p \in \mc{P}(\cN_1)$ with trace-class operators in $\cN_2$.
More precisely, consider a faithful representation of the von Neumann algebra $\cN_2$ as bounded operators on a Hilbert space $\cH_2$.\footnote{We may take $\cH_2$ to be the unique Hilbert space defined by the standard form of $\cN_2$ \cite{Haagerup1975}.}
With respect to the latter, every normal state $\sigma \in (\cN_2)_*$ is of the form $\sigma(b) = \mathrm{tr}_{\cH_2}[\rho^*_\sigma b] = \mathrm{tr}_{\cH_2}[\rho_\sigma b]$ for all $b \in \cN_2$ (Thm.~7.1.9 in \cite{KadisonRingroseII}); in physics parlance, $\rho_\sigma$ is also called a \emph{density matrix}. Throughout, we will indicate this identification in maps by writing $\phi: \cN_1 \ra \cN^*_2$ as opposed to $\widetilde{\phi}: \cN_1 \ra (\cN_2)_*$, and similarly $\varrho: \mc{P}(\cN_1) 
\ra \cN^*_2$ as opposed to $\widetilde{\varrho}: \mc{P}(\cN_1) 
\ra (\cN_2)_*$.

Under this identification we seek conditions that ensure that $\phi$ is not only positive but also completely positive.
We will achieve this in two steps: first, we extend the constraints in Gleason's theorem to dilations in Sec.~\ref{sec: dilation}; 
second, we 
enforce a consistency condition with respect to dynamical correspondences in Sec.~\ref{sec: time orientation}. Finally, in Sec.~\ref{sec: main result} we prove our main result: a generalisation of Gleason's theorem to composite systems. Along the way we show that neither of these additional constraints jeopardises applicability in the single system case. Indeed, we prove sharper versions of Gleason's theorem, thereby addressing a number of subtleties that only arise for composite systems.

\subsection{Dilations.}\label{sec: dilation}

First, we require that measures $\mu_V$ admit dilations in every commutative von Neumann subalgebra $V \in \VN$. More precisely, by Gelfand duality we may identify every context $V \in \VN$ with a compact Hausdorff $X$ space such that $\mu_V$ becomes a 
measure on $X$. In particular, we identify the projections $\mc{P}(V)$ with the clopen subsets of the Gelfand spectrum of $V$. In this way, we can interpret $\mu_V$ as a $\mc{B}(\cH)$-valued measure for $\cH=\C$ (thus also $\mc{B}(\cH) \cong \mathbb{C}$). By Naimark's theorem \cite{Naimark1943,Stinespring1955}, the latter admits a spectral dilation of the form $\mu_V = v^*\varphi_V v$, where $\varphi_V: \mc{P}(V) \hookrightarrow \mc{P}(\cK)$ is an embedding for some Hilbert space $\cK$ and $v: \C \ra \cK$ is a linear map (equivalently, a vector $v \in \mc{K}$).
By choosing $\cK$ sufficiently large, we can choose $\cK$ independent of contexts $V \in \VN$. Moreover, we take $v$ to be independent of contexts $V \in \VN$, since we can absorb any context dependence on $v$ into $\varphi_V$.

More specifically, by a finitely (completely) additive embedding we mean $\varphi_V(0) = 0$, $\varphi_V(1-p) = 1 - \varphi_V(p)$, and $\varphi_V(p_1 + p_2) = \varphi_V(p_1) + \varphi_V(p_2)$ for all $p_1,p_2 \in \mc{P}(V)$, $p_1p_2=0$ ($\varphi_V(\sum_{i\in I} p_i) = \sum_{i\in I} \varphi_V(p_i)$ for every family of pairwise
orthogonal projections $(p_i)_{i \in I}$, $p_i \in \PN$). We further require that embeddings preserve spatial tensor products, i.e., whenever $V = V_1 \bar{\otimes} V_2$, then $\varphi_V= \varphi_{V_1} \bar{\otimes} \varphi_{V_2}$, were $\varphi_{V_1}: \mc{P}(V_1) \ra \mc{P}(\cK_1)$ and $\varphi_{V_2}: \mc{P}(V_2) \ra \mc{P}(\cK_2)$ are embeddings and $\cK = \cK_1 \otimes \cK_2$. Taken together, this suggests to strengthen Def.~\ref{defn: probabilistic presheaf} as follows.

\begin{definition}\label{def: dilated probabilistic presheaf}
    Let $\cN$ be a von Neumann algebra with context category $\VN$.
    The \emph{(normal) dilated probabilistic presheaf} $\PPi_D$ of $\cN$ over $\VN$ 
    is the presheaf given
    \begin{itemize}
    	\item [(i)] on objects: for all $V\in\VN$, let
    	\begin{align*}
    			&\PPi_D(V):=\{\mu_V:\PV\ra [0,1] \mid
    		    \mu_V \text{a probability measure of the form } \mu_V = v^* \varphi_V v\text{ with}\\
    		    &\ v: \C \ra \cK \text{ linear},
    			\text{ and } \varphi_V: \mc{P}(V) \hookrightarrow \mc{P}(\mc{K}) \text{ a finitely (completely) additive embedding}\}\; ,
    	\end{align*}
    	\item [(ii)] on arrows: for all $V,\tilde{V}\in\VN$, if $\tilde{V}\subseteq V$, let
    	\begin{equation*}
    	    \PPi_D(i_{\tilde{V}V}): \PPi_D(V) \ra \PPi_D(\tilde{V}) \text{ with } \varphi_V \mapsto \varphi_V|_{\tilde{V}}\; .
    	\end{equation*}
    \end{itemize}
\end{definition}

As in Def.~\ref{defn: probabilistic presheaf}, the condition that $\varphi_V$ preserves orthogonality is required in commutative von Neumann subalgebras only. Consequently, Def.~\ref{def: dilated probabilistic presheaf} only requires global sections of $\PPi_D$ to be quasi-linear. The key difference to the probabilistic presheaf $\PPi$ in Def.~\ref{defn: probabilistic presheaf} is that we require the constraints to hold also with respect to dilations.

If Def.~\ref{def: dilated probabilistic presheaf} is to generalise the single system case, it should be consistent with it. We therefore start by analysing global sections of the dilated probabilistic presheaf in the original setting. We have the following refined version of Thm.~\ref{thm: Gleason in context form}.

\begin{theorem}\label{thm: Gleason in context form II}
    \textbf{\emph{(Gleason in contextual form (II))}} Let $\cN$ be a von Neumann algebra with no summand of type $I_2$. There is a bijective correspondence between (normal) states on $\cN$ and global sections of the (normal) dilated probabilistic presheaf $\PPi_D$ of $\cN$ over $\VN$.
\end{theorem}

\begin{proof}
    Let $\sigma \in \mc{S}(\cN)$  be (normal) state and consider the corresponding representation $\pi: \cN \ra \mc{B}(\cK)$ arising via the GNS construction from the inner product $(a,b) := \sigma(a^*b)$. In this representation, $\sigma$ is a vector state $\sigma(a) = (\pi(a)v,v)$ with $v \in \cK$ for all $a \in \cN$. Interpreting $v: \mathbb{C} \ra \cK$ as a linear map, this reads $\sigma(a) = v^* \pi(a) v$ for all $a \in \cN$, and since $\pi$ is a representation, $\varphi := \pi|_{\PN}$ is an orthomorphism. It follows that $(\sigma|_{\mc{P}(V)})_{V\in\VN}$ defines a global section of the (normal) dilated probabilistic probabilistic $\PPi_D(\VN)$.
    
    Conversely, recall that $\gamma = (v^* \varphi_V v)_{V \in \VN} \in \Gamma[\PPi(\VN)]$ defines 
    a (normal) state $\sigma_\gamma \in \SN$ by Thm.~\ref{thm: general Gleason}. In fact, 
    the embeddings $(\varphi_V)_{V \in \VN}$ define an orthomorphism $\varphi_\gamma: \PN \ra \mc{P}(\cK)$, which by Cor.~2 in \cite{BunceWright1993} extends to a (normal) Jordan $*$-homomorphism $\Phi_\gamma: \cN \rightarrow \mc{B}(\cK)$ such that $\sigma_\gamma = v^*\Phi_\gamma v$ and $\varphi_\gamma = \Phi_\gamma|_{\PN}$. This will become important in Lm.~\ref{lm: global sections to Jordan homos} below.
\end{proof}

Despite the fact that, by Thm.~\ref{thm: Gleason in context form} and Thm.~\ref{thm: Gleason in context form II}, global sections of probabilistic and dilated probabilistic presheaf correspond with states on $\cN$,\footnote{Since two global sections $\gamma,\gamma' \in \Gamma[\PPi_D(\VN)]$ are identified if and only if $\mu^\gamma_V = \mu^{\gamma'}_V$ for all $V \in \VN$.} the dilated probabilistic presheaf encodes more constraints than the probabilistic presheaf. 
To see this, note that we deduced linearity in Thm.~\ref{thm: linearity vs positivity} from a theorem due to Bunce and Wright in \cite{BunceWright1992}. Under the conditions of Def.~\ref{def: dilated probabilistic presheaf} a stronger version applies \cite{BunceWright1993}, leading to the following refinement of Thm.~\ref{thm: linearity vs positivity}.

\begin{lemma}\label{lm: global sections to Jordan homos}
    Let $\mc{N}_1$, $\mc{N}_2$ be von Neumann algebras with no summand of type $I_2$, $\cN = \cN_1 \bar{\otimes} \cN_2$, and let $\PPi_D(\mc{V}(\cN_1) \times \mc{V}(\cN_2))$ be the normal dilated probabilistic presheaf over the composite context category $\mc{V}(\cN_1) \times \mc{V}(\cN_2)$. Then every global section $\gamma \in \Gamma[\PPi_D(\mc{V}(\cN_1) \times \mc{V}(\cN_2))]$ uniquely extends to a normal linear functional $\sigma_\gamma \in \cN_*$, given for all $a \in \cN_1$ and $b \in \cN_2$ by
    \begin{equation}\label{eq: decomposable linear functional}
        \sigma_\gamma(a \otimes b) = \widetilde{\phi}_\gamma(a)(b)\; ,
    \end{equation}
    where $\phi_\gamma: \cN_1 \ra \cN_2$ is of the form $\phi_\gamma = w^*\Phi_\gamma w$ with $\Phi_\gamma: \cN_1 \ra \mc{B}(\cK)$ a normal Jordan $*$-homomorphism and $w: \cH_2 \ra \cK$ a bounded linear map for some Hilbert space $\cK$ and $\cN_2 \subset \mc{B}(\cH_2)$.
\end{lemma}


\begin{proof}
    \input{Proof_Decomposability}
\end{proof}

Recall that a map $\phi: \cN \ra \BH$ is called \emph{decomposable} if it is of the form $\phi = v^* \Phi v$ for $v: \cH \ra \cK$ a bounded linear map and $\Phi: \cN \ra \mc{B}(\cK)$ a Jordan $*$-homomorphism. Such maps satisfy a weaker notion of positivity than complete positivity \cite{Stormer1982}, owing to their close resemblance with completely positive maps $\phi = v^* \Phi v$, for which $\Phi$ is a $C^*$-homomorphism by the classification due to Stinespring \cite{Stinespring1955}.
Importantly, by a generalisation of Choi's theorem \cite{Choi1975} (see Lm.~\ref{lm: positivity vs complete positivity} below), completely positive maps $\phi_\gamma$ correspond with positive linear functionals under Eq.~(\ref{eq: decomposable linear functional}). In contrast, linear functionals $\sigma_\gamma$ corresponding to decomposable maps under Eq.~(\ref{eq: decomposable linear functional}) are generally not positive. We infer that for $\sigma_\gamma$ to be a state, $\phi_\gamma$ in Lm.~\ref{lm: global sections to Jordan homos} further needs to lift to a completely positive map, equivalently the Jordan $*$-homomorphism in Lm.~\ref{lm: global sections to Jordan homos} needs to lift to a $C^*$-homomorphism.

\subsection{Dynamical correspondences.}\label{sec: time orientation}

We recall some basic facts about Jordan algebras. A \emph{Jordan algebra} $\cJ$ is a commutative algebra, which satisfies the characteristic equation $(a^2 \circ b) \circ a = a^2 \circ (b \circ a)$ for all $a,b \in \cJ$. $\cJ$ is called a \emph{JB algebra}, if it is also a Banach space such that $||a \circ b|| \leq ||a|| \cdot ||b||$ and $||a^2|| = ||a||^2$ for all $a,b \in \cJ$. Finally, $J$ is called a \emph{JBW algebra} if it is a JB algebra with a (unique) predual. A Jordan homomorphism $\Phi: \cJ_1 \rightarrow \cJ_2$ is a linear map such that $\Phi(a \circ b) = \Phi(a) \circ \Phi(b)$ for all $a,b \in \cJ_1$. For a more details on Jordan algebras, see \cite{McCrimmon_ATasteOfJordanAlgebras}.

Note that the self-adjoint part of a von Neumann algebra $\cN$ naturally gives rise to a real JBW-algebra under the symmetrised product $a \circ b = \frac{1}{2}\{a,b\} = \frac{1}{2}(ab + ba)$. We denote this algebra by $\JNsa = (\cN_\mathrm{sa},\circ)$ and by $\JN$ its complexification.\footnote{A JBW algebra isomorphic to a subalgebra of $\JNsa$ is also called a JW algebra.} In this case we speak of Jordan $*$-homomorphisms if $\Phi$ is a Jordan homomorphism and $\Phi^*(a) = \Phi(a^*)$.

In general, $\JN$ does not determine $\cN$ completely, as it lacks compatibility with the antisymmetric part or \textit{commutator} of the associative product in $\cN$, $[a,b] = ab - ba$. In particular, for the Jordan $*$-homomorphism $\Phi$ in Lm.~\ref{lm: global sections to Jordan homos} to lift to a $C^*$-homomorphism it needs to preserve commutators. On the level of JBW algebras, this can be expressed in terms of one-parameter groups of Jordan automorphisms $\R \ni t \mapsto e^{t\delta(a)} \in \mathrm{Aut}(\cJ)$, where $\delta \in \mc{OD}_s(\cJ)$ is called a \emph{skew order derivation}, i.e., a bounded linear map $\delta: \cJ_+ \ra \cJ_+$ acting on the positive cone of $\cJ$ such that $\delta(1) = 0$ (for details, see \cite{AlfsenShultz1998a}). For $\cJ = \JNsa$, these are of the form $\delta_{ia}(b) = \frac{i}{2}[a,b]$ for all $a,b \in \JNsa$. This yields a canonical map $\psi_\cN: \JNsa \ra \mc{OD}_s(\JNsa)$, $\psi_\cN: a \ra \delta_{ia}$, which gives expression to the double role of self-adoint operators as observables and generators of symmetries \cite{AlfsenShultz1998,Baez2020}. In particular, $e^{t\psi_\cN}$ expresses the unitary evolution of the system described by $\JNsa$. We also define the map $\psi^*_{\cN} := * \circ \psi_\cN: a \ra \delta_{-ia}$ for all $a \in \JNsa$ (cf. Prop.~15 in \cite{AlfsenShultz1998a}).

Note that $\psi_\cN(a)(a) = 0$ and $[\delta_a,\delta_b] = - [\psi_\cN(a),\psi_\cN(b)]$ for all $a,b \in \JNsa$, where $\delta_a(b) := a \circ b$. More generally, Alfsen and Shultz define a \emph{dynamical correspondence} to be any map $\psi: \cJ \ra \mc{OD}_s(\cJ)$ satisfying these properties and prove that a JBW algebra $\cJ$ is a JW algebra if and only if $\cJ$ admits a dynamical correspondence.\footnote{We add that, unlike a JBW algebra, a von Neumann algebra is generally not anti-isomorphic to itself \cite{Connes1975}.} In this case, associative products bijectively correspond with dynamical correspondences (Thm.~23 in \cite{AlfsenShultz1998}). We conclude that a von Neumann algebra $\cN$ is determined as the pair $(\JN,\psi_\cN)$ with the canonical dynamical correspondence $\psi_\cN: a \ra \delta_{ia} = \frac{i}{2}[a,\cdot]$ for all $a \in \JNsa$. Moreover, we set $\cN^* := (\JN,\psi^*_\cN)$ for the associative algebra whose order of composition is reversed (under the adjoint map) with respect to the order of composition in $\cN \cong (\JN,\psi_\cN)$.

It follows that for a Jordan $*$-homomorphism $\Phi: \cJ(\cN_1) \ra \cJ(\cN_2)$ to lift to a homomorphism between the respective von Neumann algebras $\cN_1$ and $\cN_2$ is for it to preserve the respective dynamical correspondences $\psi_{\cN_1}$ and $\psi_{\cN_2}$.
Comparing with the Jordan $*$-homomorphism $\Phi_\gamma$ in Lm.~\ref{lm: global sections to Jordan homos}, we need to lift the dynamical correspondence $\psi_{\cN_2}$ to (a subalgebra of) $\mc{B}(\cK)$. To do so, note first that the argument reduces to factors in $\mc{N}_2$ \cite{AlfsenShultz1998}, 
and further that we may choose $v$ in Def.~\ref{def: dilated probabilistic presheaf} so that it preserves the factor decomposition. With this choice, let $\cN^v_2 \subseteq \mc{B}(\cK)$ be the largest von Neumann algebra which restricts to $\cN_2$ under $v$, i.e., $\cN_2 = v^* \cN_2^v v$. Clearly, $\Phi(\mc{N}_1) \subseteq \cN^v_2 \subseteq \mc{B}(\mc{K})$, and it is easy to see that in this case $\psi_{\cN_2}$ defines a unique dynamical correspondence $\psi'_{\cN_2}$ on $\cN^v_2$ 
via the linear embedding $\mc{N}_2 \hookrightarrow \cN^v_2$ induced by $v$. This entitles us to the following key definition (cf. \cite{FrembsDoering2022a}).

\begin{definition}\label{def: time-oriented global sections}
    Let $\cN_1 = (\cJ(\cN_1),\psi_{\cN_1})$, $\cN_2 = (\cJ(\cN_2),\psi_{\cN_2})$ be von Neumann algebras. A global section of the dilated probabilistic presheaf $\gamma \in \Gamma[\PPi_D(\mc{V}(\cN_1) \times \mc{V}(\cN_2))]$ is called \emph{time-oriented} if the Jordan $*$-homomorphism $\Phi_\gamma$ in Lm.~\ref{lm: global sections to Jordan homos} preserves dynamical correspondences,
    \begin{equation}\label{eq: Jordan to vN condition}
        \Phi_\gamma \circ \psi^*_{\cN_1} = \psi'_{\cN_2} \circ \Phi_\gamma\; ,
    \end{equation}
    where $\psi'_{\cN_2}$ denotes the dynamical correspondence on $\cN_2^v \subset \mc{B}(\cK)$ uniquely defined by $\psi_{\cN_2}$.
\end{definition}

The terminology reflects the fact that the one-parameter groups $\mathbb{R} \ni t \mapsto e^{t\psi_\cN}$ 
express unitary dynamics in quantum theory. Consequently, they give physical meaning to $t$ as a time parameter. In particular, $\psi_\cN: a \ra \delta_{ia}$ fixes the (canonical) forward time direction of the system described by $\cN$, whereas $\psi^*_{\cN} := * \circ \psi_\cN: a \ra \delta_{-ia}$
describes the system under time reversal \cite{Doering2014} (see also \cite{FrembsDoering2022a,Frembs2022a}). We emphasise that the appearance of $\psi^*_1$ and $\psi_2$ in Def.~\ref{def: time-oriented global sections} is a consequence of interpreting the normal linear functional $\sigma_\gamma: \cN_1 \bar{\otimes} \cN_2 \ra \C$ in Lm.~\ref{lm: global sections to Jordan homos} as a linear map $\phi_\gamma: \cN_1 \ra \cN^*_2$ (under the identification between trace-class operators and normal linear functionals), equivalently $\phi_\gamma: \cN^*_1 \ra \cN_2$,
where the order of composition in $\cN^*_1$ is naturally reversed (under the adjoint map) with respect to the order of composition in $\cN_1$.

Once again, we show that Def.~\ref{def: time-oriented global sections} poses no additional constraint in the single system case.

\begin{theorem}\label{thm: Gleason in context form III}
    \textbf{\emph{(Gleason in contextual form (III))}} Let $\cN = (\cJ(\cN),\psi_\cN)$ be a von Neumann algebra with no summand of type $\text{I}_2$. There is a bijective correspondence between (normal) states on $\cN$ and time-oriented global sections of the (normal) dilated probabilistic presheaf $\PPi_D$ of $\cN$ over $\VN$.
\end{theorem}

\begin{proof}
    By Thm.~\ref{thm: Gleason in context form II}, states on $\cN$ correspond with global sections of $\PPi_D(\VN)$. It is thus sufficient to show that every state is already time-oriented.
    This follows from the fact that every positive linear functional is automatically completely positive \cite{Stinespring1955}. More precisely, recall that the GNS construction for $\sigma \in \SN$ trivially yields a dilation of the form in Lm.~\ref{lm: global sections to Jordan homos} (for $\cN_2 = \C$), namely $\sigma(a) = v^*\pi(a)v = \mathrm{tr}_{\cK}[vv^*\pi(a)] = \mathrm{tr}_{\cK}[(vv^*)^*\pi(a)]$ for all $a \in \cN$, where $v: \C \ra \mc{K}$ and $\pi: \cN \ra \mc{B}(\cK)$ is a representation of $\cN$.\footnote{Note that this construction works for all states $\sigma \in \mc{S}(\cN)$, not just for normal states.}
    In particular, $\sigma$ is time-oriented with respect to the algebras $\cN_1 = \cN$ and $\cN_2 = \C$.
\end{proof}

\subsection{Gleason's theorem for composite systems.}\label{sec: main result}

The refinements of Thm.~\ref{thm: general Gleason} in Thm.~\ref{thm: Gleason in context form II} and Thm.~\ref{thm: Gleason in context form III} both reveal extra structure in the classification of states in the single system case. In turn, having identified this additional structure in Def.~\ref{def: dilated probabilistic presheaf} and Def.~\ref{def: time-oriented global sections}, we may impose it in order to generalise Gleason's theorem to composite systems. Then, every $\gamma \in \Gamma[\PPi_D(\mc{V}(\cN_1) \times \mc{V}(\cN_2))]$ yields a map $\phi_\gamma(a) = v^*\Phi_\gamma(a)v$ with $\Phi_\gamma$ a $C^*$-homomorphism, which implies that $\phi_\gamma$ is completely positive by Stinespring's theorem. In the final step, we deduce from this the existence of a positive linear functional $\sigma_\gamma: \cN \ra \C$ with $\cN = \cN_1 \bar{\otimes} \cN_2$. In finite dimensions, this follows from Choi's theorem \cite{Choi1975}. However, the latter hinges on the existence of a finite trace. For the general case, we use the following generalisation based on a result by Belavkin and Staszewski \cite{Belavkin1986}.\footnote{If $\phi: \cN \ra \cN$, where $\cN$ has a cyclic and separating vector, the result also follows from \cite{Stormer2014}.} 

\begin{lemma}\label{lm: positivity vs complete positivity}
    The normal linear functional $\sigma_\gamma(a \otimes b) = \widetilde{\phi}_\gamma(a)(b)$ in Lm.~\ref{lm: global sections to Jordan homos} is positive if and only if the 
    linear map $\phi_\gamma: \cN^*_1 = (\mc{J}(\cN_1),\psi^*_{\cN_1}) \ra \cN_2 = (\mc{J}(\cN_2),\psi_{\cN_2})$ is completely positive.
\end{lemma}

\begin{proof}
\input{Proof_PositivityVsCompletePositivity}
\end{proof}

We finally arrive at a generalisation of Gleason's theorem to composite systems.

\begin{theorem}\label{thm: composite Gleason}
     Let $\cN_1 = (\cJ(\cN_1),\psi_{\cN_1})$, $\cN_2 = (\cJ(\cN_2),\psi_{\cN_2})$ be von Neumann algebras with no summands of type $I_2$ and let $\cN = \cN_1 \bar{\otimes} \cN_2$. Then there is a bijective correspondence between the set of normal states on $\cN$ and the set of time-oriented global sections of the normal dilated probabilistic presheaf $\PPi_D$ over $\mc{V}(\cN_1) \times \mc{V}(\cN_2)$.
\end{theorem}

\begin{proof}
    \input{Proof_KeyTheorem}
\end{proof}

\vspace{-0.2cm}

We finish with a few remarks. Note that the key steps in the argument leading to Thm.~\ref{thm: composite Gleason} apply to $C^*$-algebras. The crucial exception is lifting quasi-linearity to linearity. As such the restriction to von Neumann algebras might not be optimal. For instance, our result may be strengthened to $AW^*$-algebras, for which Gleason's theorem holds by \cite{Hamhalter2015}. 

Comparing with $C^*$-algebras, we further remark that we restricted to normal states, since non-normal product states generally do not have a unique extension in the spatial tensor product.
On the other hand, in the spatial tensor product of $C^*$-algebras every product state has a unique extension (Prop.~11.1.1 in \cite{KadisonRingroseII}). This suggests that for $AW^*$-algebras the bijective correspondence in Thm.~\ref{thm: composite Gleason} further extends beyond normal states.

Moreover, we observe that Thm.~\ref{thm: composite Gleason} can be formulated in an intrinsically Jordan-algebraic setting. Note that Bunce and Wright prove a generalisation of Gleason's theorem for JBW-algebras in \cite{BunceWright1985}. The case considered here is the one where JBW algebras arise as self-adjoint parts of von Neumann algebra, hence, are JW algebras. Alfsen and Shultz show that this is the case if and only if the JBW algebra admits a dynamical correspondence \cite{AlfsenShultz_StateSpaceOfOperatorAlgebras}. Thm.~\ref{thm: composite Gleason} applies for any choice of dynamical correspondences on the factor JW algebras.

In turn, in the context of general JBW algebras, the concept of state requires a weaker notion of positivity, reflecting the generalisation from completely positive to decomposable maps \cite{Stormer1982}. Clearly, the relation with (preservation of) dynamical correspondences in Def.~\ref{def: time-oriented global sections} is lost in this case. We point out that this is closely analogous to the limited applicability of Tomita-Takesaki theory in JBW algebras \cite{HaagerupHanche-Olsen1984}.

\section{Conclusion}

We proved a generalisation of Gleason's theorem to composite systems. An ad hoc attempt establishes a correspondence with normal linear functionals on the composite algebra that are positive on product operators, but fails to single out those that are positive (Thm.~\ref{thm: linearity vs positivity}). We remedy this by strengthening the additivity constraints on measures in commutative von Neumann subalgebras to dilated systems, together with a consistency condition between dynamical correspondences on the respective subalgebras (Thm.~\ref{thm: composite Gleason}). Neither of these conditions changes the result in the single system case (Thm.~\ref{thm: Gleason in context form II} and Thm.~\ref{thm: Gleason in context form III}), since positive linear functionals are completely positive as a consequence of Stinespring's theorem \cite{Stinespring1955}.




Apart from its mathematical value, our result also carries physical significance. Gleason's theorem (Thm.~\ref{thm: general Gleason}) plays a crucial role in the foundations of quantum theory, where it justifies Born's rule. It is all the more interesting that a similar result extends to composite systems. For instance, we note that within the framework of algebraic quantum field theory the observable algebra is naturally composed of local algebras \cite{HaagKastler1964}.

Thm.~\ref{thm: composite Gleason} was foreshadowed for finite type I factors, and with a focus on
classifying quantum from non-signalling correlations in \cite{FrembsDoering2022a}. Its close connection with entanglement classification will be discussed elsewhere \cite{Frembs2022a}. For a broader perspective on the significance of contextuality in the foundations of quantum mechanics we refer to \cite{DoeringFrembs2019a}.\\


\paragraph*{Acknowledgements.} This work is supported through a studentship in the Centre for Doctoral Training on Controlled Quantum Dynamics at Imperial College funded by the EPSRC, by grant number FQXi-RFP-1807 from the Foundational Questions Institute and Fetzer Franklin Fund, a donor advised fund of Silicon Valley Community Foundation, and ARC Future Fellowship FT180100317.

\clearpage
\bibliographystyle{siam}
\bibliography{bibliography}

\end{document}

%% file: Proof_Linearity.tex
Clearly, every normal product state $\sigma = \sigma_1 \bar{\otimes} \sigma_2 \in \cN_*$ yields a product global section $\gamma_\sigma = \gamma_{\sigma_1} \times \gamma_{\sigma_2} \in \Gamma[\PPi(\mc{V}(\cN_1) \times \mc{V}(\cN_2))]$, as a consequence of Thm.~\ref{thm: Gleason in context form}, applied to $\cN_1$ and $\cN_2$ individually. In particular, $\gamma(1) = \sigma(1) = 1$ by normalisation, and $\sigma(a \otimes b) \geq 0$ for all $a \in (\cN_1)_+$ and $b \in (\cN_2)_+$ by positivity of the probability measures $\mu^\gamma_V$ for all $V\in\VN$. Since every normal state can be approximated by finite linear combinations of normal product states (Prop 11.2.8, \cite{KadisonRingroseII}), the correspondence extends to all normal states $\sigma \in \cN_*$.
    
Conversely, we show that every global section $\gamma \in \Gamma[\PPi(\mc{V}(\cN_1) \times \mc{V}(\cN_2))]$ defines a unique normal linear functional $\sigma_\gamma$ on $\cN = \cN_1 \bar{\otimes} \cN_2$, which restricts to $\gamma$ on $\mc{V}(\cN_1) \times \mc{V}(\cN_2)$, i.e., $\gamma = (\sigma|_{\mc{P}(V)})_{\mc{V}(\cN_1) \times \mc{V}(\cN_2)}$. Fix a context $V_1 \in \mc{V}(\mc{N}_1)$ and consider the corresponding partial order of contexts under inclusion, inherited from $\mc{V}_{1\&2} := \mc{V}(\mc{N}_1) \times \mc{V}(\mc{N}_2)$ by restriction,
\begin{equation*}
    \mc{V}_{1\&2}(V_1) := \{V_1 \times V_2 \mid V_2 \in \mc{V}(\mc{N}_2)\}\; .
\end{equation*}
In every context $V = V_1 \bar{\otimes} V_2 \in \mc{V}_{1\&2}$, the probability measure $\mu^\gamma_V \in \PPi(\mc{V}_{1\&2})_V$ corresponding to the global section $\gamma$ can be written (in terms of conditional probabilities)
as follows:
\begin{equation}\label{eq: local Gleason}
    \forall p \in \mc{P}(V_1), q \in \mc{P}(V_2):\ \mu^\gamma_V(p,q) = \mu_{V_1}^\gamma(p) \mu_{V_2}^\gamma(q \mid p) = \mu_{V_1}^\gamma(p) \gamma_2^p(q) = \mu_{V_1}^\gamma(p) \sigma_2^p(q)\; .
\end{equation}
Here, $(\mu_{V_2}^{\gamma_2} (\cdot \mid p))_{V_2 \in \mc{V}(\mc{N}_2)} =: \gamma_2^p \in \Gamma[\PPi(\mc{V}_{1\&2}(V_1))]$ is a global section of the probabilistic presheaf $\PPi(\mc{V}_{1\&2}(V_1))$, which also depends on $p \in \mc{P}(V_1)$. Since $\mc{V}_{1\&2}(V_1) \cong \mc{V}(\mc{N}_2)$, such global sections correspond with normal states on $\cN_2$ 
by Thm.~\ref{thm: Gleason in context form}, hence, $\sigma_2^p \in (\cN_2)_*$ for all $p \in \mc{P}(V_1)$. Moreover, since $V_1 \in \mc{V}(\mc{N}_1)$ was arbitrary, Eq.~(\ref{eq: local Gleason}) holds for all $p \in \mc{P}(\cN_1)$.

Define a map $\widetilde{\varrho}_\gamma: \mc{P}(\cN_1) \ra (\cN_2)_*$ by $\widetilde{\varrho}_\gamma(p) := \mu_{V_1}^\gamma(p) \sigma_2^p$. Clearly, $\widetilde{\varrho}_\gamma$ is bounded and $(\cN_2)_*$ is a Banach space (as a closed subspace of the continuous dual of $\mc{N}_2$). Moreover, for $p = p_1 + p_2$ with $p_1,p_2 \in \mc{P}(\mc{N}_1)$ orthogonal, i.e., $p_1p_2 = 0$,
\begin{equation*}\label{eq: finite additive operator measure}
    \widetilde{\varrho}_\gamma(p) = \mu_{V_1}^\gamma(p) \sigma_2^p
    = \mu_{V_1}^\gamma(p_1)\sigma_2^{p_1} + \mu_{V_1}^\gamma(p_2)\sigma_2^{p_2}
    = \widetilde{\varrho}_\gamma(p_1) + \widetilde{\varrho}_\gamma(p_2)\; ,
\end{equation*}
by additivity of $\gamma$. Similarly, $\widetilde{\varrho}_\gamma$ inherits complete additivity from $\gamma$.
By Thm.~A in \cite{BunceWright1992}, $\widetilde{\varrho}_\gamma$ uniquely extends to a normal linear map $\widetilde{\phi}_\gamma: \cN_1 \ra (\cN_2)_*$. Hence, $\sigma_\gamma(p \otimes q) := \widetilde{\phi}_\gamma(p)(q)$ is a normal linear functional such that $\sigma_\gamma(1) = \gamma(1) = 1$ and $\sigma_\gamma(a \otimes b) \geq 0$ for all $a \in (\cN_1)_+$, $b \in (\cN_2)_+$ by positivity of the (product) measures $\mu^\gamma_V$ for all $V\in\mc{V}(\cN_1)\times\mc{V}(\cN_2)$.

%% file: Proof_Decomposability.tex
Recall that in the proof of Thm.~\ref{thm: linearity vs positivity} we lifted the map $\widetilde{\varrho}_\gamma: \mc{P}(\cN_1) \ra (\cN_2)_*$ to a positive bounded linear map $\widetilde{\phi}_\gamma: \cN_1 \ra (\cN_2)_*$ (cf. \cite{BunceWright1992}). Similarly, for $\gamma \in \Gamma[\PPi_D(\mc{V}(\cN_1) \times \mc{V}(\cN_2))]$ we construct a completely additive measure $\widetilde{\varrho}_\gamma: \mc{P}(\cN_1) \ra (\cN_2)_*$, which by the correspondence between normal states and positive trace-class operators (Thm.~7.1.9 in \cite{KadisonRingroseII}) yields a map $\varrho_\gamma: \mc{P}(\cN_1) \ra (\cN_2)_+$ into positive trace-class operators (viz. density matrices). 

Furthermore, note that by assumption there exists a Hilbert space $\cK = \cK_1 \otimes \cK_2$, a map $v: \C \ra \cK$ and orthomorphisms $\varphi_1: \mc{P}(\cN_1) \ra \mc{P}(\cK_1)$ and $\varphi_2: \mc{P}(\cN_2) \ra \mc{P}(\cK_2)$ (uniquely defined by $\varphi_1|_{V_1} = \varphi_{V_1}$ and $\varphi_2|_{V_2} = \varphi_{V_2}$ for all $V_1 \in \mc{V}(\cN_1)$, $V_2 \in \mc{V}(\cN_2)$, respectively) such that
\begin{align*}
    \sigma_\gamma(p\otimes q)
    = v^*(\varphi_1(p) \otimes \varphi_2(q))v
\end{align*}
for all $p \in \mc{P}(\cN_1)$ and $q \in \mc{P}(\cN_2)$. Without loss of generality, we may identify $\cK_2 = \cH_2$ (for $\cN_2 \subset \mc{B}(\cH_2)$) and take $\varphi_2(q) = q$ to be the identity.
Then let $v = \sum_{ij} c_{ij} v_{1i} \otimes v_{2j} \in \cK$, where $c_{ij} \in \C$ and $\{v_{1i}\}_i$, $\{v_{2j}\}_j$ are orthonormal bases in $\cK_1$ and $\cK_2$, respectively. 
Since $\sigma_\gamma(p \otimes \cdot)$ corresponds to a normal state on $\cN_2$ for all $p \in \mc{P}(\cN_1)$ (by Thm.~\ref{thm: general Gleason}), we write
\begin{align*}
    \sigma_\gamma(p\otimes q)
    &= (\sum_{ij} c^*_{ij} v^*_{1i} \otimes v^*_{2j})\ (\varphi_1(p) \otimes q)\ 
    (\sum_{kl} c_{kl} v_{1k} \otimes v_{2l})\\
    &= \sum_{ijkl} c^*_{ij}c_{kl}\ (v^*_{1i}\varphi_1(p)v_{1k})\ (v^*_{2j}qv_{2l})\\
    &= \mathrm{tr}_{\cH_2}\bigg[\bigg(\sum_{ijkl} c_{ij}c^*_{kl}\ (v^*_{1k}\varphi_1(p)v_{1i})\ v_{2j}v^*_{2l}\bigg)^*q\bigg]\; .
\end{align*}
where we used that $\varphi^*_1(p) = \varphi_1(p)$ for all $p \in \mc{P}(\cN_1)$ in the last step. Note that under the Hermitian adjoint in the last line the inner product on $\cN_1$ is changed to its complex conjugate. To reflect this, let $\mc{B}(\cK_1)^*$ denote the algebra $\mc{B}(\cK_1)$ under the Hermitian adjoint, such that for all pairs $v_{1i},v_{1k} \in \cK_1$, $\tilde{v}_{1i},\tilde{v}_{1k} \in \widetilde{\cK}_1 \cong \cK^*_1$ (the dual of $\cK_1$) and $a \in \mc{B}(\cK_1)$ we have
\begin{equation}\label{eq: inner product N_1}
    (\tilde{v}^*_{1k} a^* \tilde{v}_{1i})_{\mc{B}(\cK_1)^*}
    := (v^*_{1k} a v_{1i})^*_{\mc{B}(\cK_1)}\; .
\end{equation}
Define a bounded linear map $w: \cK_2 \ra \widetilde{\cK}_1 \otimes \cK_2$ by $w = \sum_{ij} c^*_{ij} \tilde{v}_{1i} \otimes v_{20}v^*_{2j}$,
and let $\varphi' = \varphi_1 \otimes 1_{\mc{B}(\cK_2)}$. 
This yields an orthomorphism $\varphi': \mc{P}(\cN_1) \ra \mc{P}(\cK)$ such that 
\begin{align*}
    w^*\varphi'(p)w
    &= (\sum_{ij} c_{ij} \tilde{v}^*_{1i} \otimes v_{2j}v^*_{20})\ \varphi'(p)\ (\sum_{kl} c^*_{kl} \tilde{v}_{1k} \otimes v_{20}v^*_{2l})\\
    &= \sum_{ijkl} c_{ij}c^*_{kl}\ (\tilde{v}^*_{1i} \varphi_1(p)\tilde{v}_{1k})_{\mc{B}(\cK_1)^*}\ v_{2j}v^*_{2l}\\
    &= \sum_{ijkl} c_{kl}c^*_{ij}\ (\tilde{v}^*_{1k} \varphi_1(p)\tilde{v}_{1i})_{\mc{B}(\cK_1)^*}\ v_{2l}v^*_{2j} \\
    &= \left(\sum_{ijkl} c_{ij}c^*_{kl}\ (v^*_{1k}\varphi_1(p)v_{1i})_{\mc{B}(\cK_1)}\ v_{2j}v^*_{2l}\right)^*\; ,
\end{align*}
where we changed labels $i \leftrightarrow k$, $j \leftrightarrow l$ in the third line, and used Eq.~(\ref{eq: inner product N_1}) together with $\varphi^*_1(p) = \varphi_1(p)$ for all $p \in \mc{P}(\cN_1)$ in the last line. Consequently,
\begin{equation*}
    \widetilde{\varrho}_\gamma(p)(q) = \mathrm{tr}_{\cH_2}[w^*\varphi'(p)wq]\; .
\end{equation*}
Finally, by Cor.~1 in \cite{BunceWright1993} the completely additive orthomorphism $\varphi': \mc{P}(\cN_1) \ra \mc{P}(\cK)$ extends to a normal Jordan $*$-homomorphism $\Phi_\gamma: \mc{J}(\cN_1) \ra \mc{B}(\cK)$. 

%% file: Proof_PositivityVsCompletePositivity.tex
By Lm.~\ref{lm: global sections to Jordan homos}, $\phi_\gamma: \cN^*_1 \ra \cN_2$ is a decomposable map (between Jordan algebras $\mc{J}(\cN_1)$ and $\mc{J}(\cN_2)$) \cite{Stormer1982}. Assume that $\phi_\gamma$ is moreover completely positive.
Fix faithful, normal, semi-finite weights $\omega_1: \mc{N}_1 \rightarrow \mathbb{C}$
and $\omega_2: \mc{N}_2 \rightarrow \mathbb{C}$ with respective faithful GNS representations $\pi_{\omega_1}: \cN_1 \ra \mc{B}(\cH_1)$ and $\pi_{\omega_2}: \cN_2 \ra \mc{B}(\cH_2)$. In the product representation $\pi_\omega: \cN_1 \bar{\otimes} \cN_2 \ra \mc{B}(\cH_1\otimes\cH_2)$ induced by the weight $\omega = \omega_1 \bar{\otimes} \omega_2$ we have $\sigma_\gamma(a\otimes b) = \omega(\rho^*_\gamma(a\otimes b))$ for all $a \in \cN_1$, $b \in \cN_2$ and some trace-class operator $\rho_\gamma \in \cN_1 \bar{\otimes} \cN_2 \subset \mc{B}(\cH_1\otimes\cH_2)$.\footnote{In other words, we again identify $\cN_1 \subset \mc{B}(\cH_1)$ and $\cN_2 \subset \mc{B}(\cH_1)$ with their standard forms \cite{Haagerup1975}.} We show that $\rho_\gamma$ is 
positive and self-adjoint.

Define a normal, completely positive reference map $\phi_0: \cN^*_1 \ra \cN_2 \subset \mc{B}(\cH_2)$ by $\phi_0 = \omega^*_1 1_{\mc{B}(\cH_2)} 
= v_0^* \pi_0 v_0$, where $v_0 = v_1 \otimes 1_{\mc{B}(\cH_2)}: \cH_2 \ra \cH_1 \otimes \cH_2$ for $v_1: \C \ra \cH_1$ in $\omega^*_1 = (v^*_1\pi_{\omega_1}v_1)^* = v^*_1 \pi^*_{\omega_1} v_1 = v^*_1 (\pi_{\omega_1} \circ *) v_1$ (cf. \cite{Stinespring1955})
and $\pi_0 = (\pi_{\omega_1} \circ *) \otimes 1_{\mc{B}(\cH_2)}$. 
Recall that complete positivity of $\phi_0$ means $\sum_{i,j=1}^n (\phi_0(a_{ij})\eta_i,\eta_j) \geq 0$ for all $\eta_i,\eta_j \in \mc{H}_2$, $n \in \mathbb{N}$ whenever $a_{ij} \in M_n(\cN_1) \cong M_n(\cN^*_1)_+$. Let $(a_{ij})_m \in \mc{N}^*_1$ be a family of sequences such that
\begin{equation*}
    \mathrm{lim}_{m \rightarrow \infty} \sum_{i,j=1}^n (\phi_0((a_{ij})_m)\eta_i,\eta_j) = \sum_{i,j=1}^n (\eta_i,\eta_j)\ \mathrm{lim}_{m \rightarrow \infty} \omega_1((a_{ij})_m) = 0\; .
\end{equation*}
Since $\eta_i,\eta_j$ are arbitrary and $\omega_1$ is a faithful, normal weight, we conclude $\mathrm{lim}_{m \rightarrow \infty} (a_{ij})_m = 0$ in the ultraweak topology.
Since $\phi_\gamma$ is also a completely positive and normal map, we have $\mathrm{lim}_{m \rightarrow \infty} \sum_{i,j=1}^n (\phi_\gamma((a_{ij})_m)\eta_i,\eta_j) = 0$ for all $\eta_i, \eta_j \in \mc{H}_2$ (cf. Prop.~III.2.2.2 in \cite{Blackadar_OperatorAlgebras}). In the terminology of \cite{Belavkin1986}, $\phi_\gamma$ is thus strongly completely absolutely continuous with respect to $\phi_0$. 
Consequently, Thm.~2 in \cite{Belavkin1986} asserts the existence of a 
positive, self-adjoint operator $\rho_\gamma$ 
such that $\phi_\gamma = v_0^* \rho_\gamma \pi_0 v_0$. 
It follows that $\sigma_\gamma(a\otimes b) = \omega(\rho_\gamma(a\otimes b))$ for all $a \in \cN_1$, $b \in \cN_2$ defines a positive and thus normal state on $\cN = \cN_1 \bar{\otimes} \cN_2$.

Conversely, if $\sigma_\gamma$ is positive, there exists a positive, self-adjoint operator $\rho_\gamma$ of trace-class such that $\sigma(a \otimes b) = \omega(\rho_\gamma(a \otimes b))$ for all $a \in \cN_1$, $b \in \cN_2$, hence, $\phi_\gamma = v_0^*\rho_\gamma\pi_0 v_0$. 
In particular, $\phi_\gamma: \cN^*_1 \ra \cN_2$ is completely positive by complete positivity of $\phi_0$.

%% file: Proof_KeyTheorem.tex
By Thm.~\ref{thm: Gleason in context form II},
every normal state $\sigma \in (\cN_1 \bar{\otimes} \cN_2)_*$ restricts to a
global section of the normal dilated probabilistic presheaf, $\gamma_\sigma = (\sigma|_{\mc{P}(V)})_{V \in \mc{V}(\cN_1) \times \mc{V}(\cN_2)} \in \Gamma[\PPi_D(\mc{V}(\cN_1) \times \mc{V}(\cN_2))]$. It is thus sufficient to show that every normal state $\sigma$ can be written in the form $\sigma(a \otimes b) = \widetilde{\phi}(a)(b)$ for all $a \in \cN_1$, $b \in \cN_2$ and $\phi: \cN^*_1 \ra \cN_2$ completely positive, which follows from Lm.~\ref{lm: positivity vs complete positivity}.

Conversely, let $\gamma \in \Gamma[\PPi_D(\mc{V}(\cN_1) \times \mc{V}(\cN_2))]$. By Lm.~\ref{lm: global sections to Jordan homos}, $\gamma$ corresponds to a normal linear functional of the form $\sigma_\gamma(a \otimes b) = \widetilde{\phi}_\gamma(a)(b)$ for all $a \in \cN_1$, $b \in \cN_2$ and $\phi_\gamma$ a decomposable map (between the respective Jordan algebras $\mc{J}(\cN_1)$ and $\mc{J}(\cN_2)$), i.e., $\phi_\gamma(a) = w^*\Phi_\gamma(a)w$ for $w: \cH_2 \ra \cK$ with $\cN_2 \subset \mc{B}(\cH_2)$ and $\Phi_\gamma: \mc{J}(\cN_1) \ra \mc{J}(\mc{B}(\cK))$ a Jordan $*$-homomorphism. Finally, since $\gamma$ is time-oriented with respect to the dynamical correspondences $\psi^*_{\cN_1} = * \circ \psi_{\cN_1}$ and $\psi_{\cN_2}$, $\Phi_\gamma: \cN^*_1 \ra \cN^v_2 \subset \mc{B}(\cK)$ is also a $C^*$-homomorphism, and $\phi_\gamma$ thus completely positive by Stinespring's theorem \cite{Stinespring1955} (cf. Thm.~III.2.2.4 in \cite{Belavkin1986}). Lm.~\ref{lm: positivity vs complete positivity} asserts that $\sigma_\gamma$ is positive and thus a normal state on $\cN = \cN_1 \bar{\otimes} \cN_2$.